\documentclass{article}
\usepackage{amsfonts}
\usepackage{mathrsfs}
\usepackage{amsthm}
\usepackage{amssymb}
\usepackage{amsmath}
\usepackage{enumerate}

\theoremstyle{plain}
\newtheorem{theorem}{Theorem}[section]
\newtheorem{proposition}[theorem]{Proposition}
\newtheorem{corollary}[theorem]{Corollary}

\theoremstyle{definition}
\newtheorem{definition}{Definition}[section]

\theoremstyle{remark}
\newtheorem{remark}{\textbf{Remark}}[section]

\theoremstyle{example}
\newtheorem{example}{Example}[section]

\numberwithin{equation}{section}

\title{The uniform measure for quantum walk on hypercube:\\
a quantum Bernoulli noises approach}
\author{Ce Wang\\
Yau Mathematical Sciences Center, Tsinghua University\\
Beijing 100084, People's Republic of China}
\date{}
\begin{document}
\maketitle

\noindent\textbf{Abstract.}\ \
In this paper, we present a quantum Bernoulli noises approach to quantum walks on hypercubes.
We first obtain an alternative description of a general hypercube and then, based on the alternative description,
we find that the operators $\partial_k^* + \partial_k$ behave actually as the shift operators,
where $\partial_k$ and $\partial_k^*$ are the annihilation and creation operators acting on Bernoulli functionals, respectively.
With the above operators as the shift operators on the position space, we introduce a discrete-time quantum walk model on a general hypercube
and obtain an explicit formula for calculating its probability distribution at any time.
We also establish two limit theorems showing that the averaged probability distribution of the walk even converges to the uniform probability distribution.
Finally, we show that the walk produces the uniform measure as its stationary measure on the hypercube
provided its initial state satisfies some mild conditions. Some other results are also proven.
\vskip 2mm

\vskip 2mm


\section{Introduction}

As is known, random walks have been extensively used in modeling various physical processes and in developing random algorithms.
Recent two decades have witnessed great attention paid to quantum walks, which are quantum analogs of random walks
(see, e.g. \cite{y-aharonov, kempe-1, konno-2, portugal,venegas} and references therein).
Quantum walks are also known as quantum random walks, and have found wide application in quantum computation.
Due to the quantum interference effects, quantum walks greatly outperform random walks at certain computational tasks,
and moreover it has turned out that quantum walks constitute universal models of quantum computation \cite{portugal}.

From a perspective of mathematical physics, a quantum walk can be viewed as a quantum dynamical system (QDS) driven by a unitary operator
on the tensor space of an $l^2$-space associated with a graph and an auxiliary Hilbert space.
For instance, the well-known Hadamard walk can be thought of a QDS driven by a unitary operator on the tensor space $l^2(\mathbb{Z})\otimes \mathbb{C}^2$,
where the integer lattice $\mathbb{Z}$ is essentially a graph.
There are two basic types of quantum walks in the literature: the type of discrete-time quantum walks and the type of continuous-times quantum walks.
In this paper, we focus on the former.

Hypercubes are a special class of regular graphs, which play an important role in computer science.
Recently there has been much interest in quantum walks on hypercubes.
In their survey paper, Aharonov et al \cite{d-aharonov} presented the ground for a theory of quantum walks on graphs including hypercubes.
Moore and Russell \cite{moore} studied the instantaneous mixing time of quantum walks on the $n$-dimensional hypercube,
while Kempe \cite{kempe-2} investigated the hitting time of these walks. Alagi\'{c} and Russell \cite{alagic} considered decoherence in quantum walks on hypercubes.
There are many other works concerning quantum walks on hypercubes in the literature (see, e.g. \cite{makmal, marquezino}).

However, little attention has been paid to the problem of finding a way to produce the uniform measure on a general hypercube via a quantum walk,
which is interesting as is pointed out in \cite{konno-1}.
In this paper, as one of our main purposes, we would like to provide a solution to such problem.

Quantum Bernoulli noises (QBN) are annihilation and creation operators acting on Bernoulli functionals, which satisfy the canonical anti-commutation relations (CAR) in equal time.
It has turned out that QBN can play an important role in describing the irreversible evolution of an open quantum system (especially a Fermi system)
\cite{wang-chen, w-t-r}. In 2016, Wang and Ye introduced a quantum walk model on the $1$-dimensional integer lattice $\mathbb{Z}$ in terms of QBN
and showed its strong decoherence property \cite{w-y}.

In the present paper, as our another main purpose, we would also like to present a QBN approach to quantum walks on hypercubes. Our main work is as follows.
\begin{itemize}
  \item We obtain an alternative description of the $(n+1)$-dimensional hypercube, where $n$ is a general nonnegative integer. And based on this alternative description,
    we find that the operators $\{\partial_k^* + \partial_k \mid 0\leq k \leq n\}$ behave actually as the shift operators, where $\partial_k$ and $\partial_k^*$
    are the annihilation and creation operators acting on Bernoulli functionals.
  \item We introduce a notion of coin operator system on a general Hilbert space and obtain a necessary and sufficient condition for an operator system to be a coin operator system.
  \item With $\{\partial_k^* + \partial_k \mid 0\leq k \leq n\}$ as the shift operators on the position space, we introduce a quantum walk model on the $(n+1)$-dimensional hypercube and obtain an explicit formula for calculating the probability distribution of the walk at any time.
  \item We establish two limit theorems showing that the averaged probability distribution of the walk even converges to the uniform probability distribution on the
     vertex set of the $(n+1)$-dimensional hypercube.
  \item Finally, we prove that the walk produces the uniform measure as its stationary measure on the $(n+1)$-dimensional hypercube
     provided its initial state satisfies some mild conditions.
\end{itemize}
Our work shows that QBN can provide a way to produce the uniform measure on a general hypercube via a quantum walk.
Moreover, our work also suggests that, for a quantum walk on a hypercube, its ``components in the coin space'' can be
the determining factors of its probability distributions and evolution behavior.

The paper is organized as follows. In Section~\ref{sec-2}, we recall some necessary notions and facts about quantum Bernoulli noises.
Our main work then lies in Section~\ref{sec-3}, where we give an alternative description of the $(n+1)$-dimensional hypercube (Subsection~\ref{subsec-3-1}),
characterize our coin operator systems (Subsection~\ref{subsec-3-2}), define our quantum walk model (Subsection~\ref{subsec-3-3}),
examine the probability distribution of the walk (Subsection~\ref{subsec-3-4}), prove that the walk can produce the uniform measure as its stationary measure
(Subsection~\ref{subsec-3-5}), and offer some examples (Subsection~\ref{subsec-example}). Finally in Section~\ref{sec-4}, we make some conclusion remarks.


\vskip 2mm

\textbf{Conventions.}
Throughout this paper, by a Hilbert space we means a separable complex Hilbert space whose inner product is conjugate linear in the first variable
and linear in the second variable. If $A$ is an bounded operator on a Hilbert space, then $A^*$ denotes its adjoint.
For Hilbert spaces $\mathcal{H}_1$ and $\mathcal{H}_2$, their tensor space is written as $\mathcal{H}_1\otimes \mathcal{H}_2$.
By convention, $\dim \mathcal{H}$ means the dimension of a Hilbert space $\mathcal{H}$.

\section{Quantum Bernoulli noises}\label{sec-2}

Let $\Omega$ be the set of all functions $\omega\colon \mathbb{N} \mapsto \{-1,1\}$, and
$(\zeta_n)_{n\geq 0}$ the sequence of canonical projections on $\Omega$ given by
\begin{equation}\label{eq-2-1}
    \zeta_n(\omega)=\omega(n),\quad \omega\in \Omega.
\end{equation}
Let $\mathscr{F}$ be the $\sigma$-field on $\Omega$ generated by the sequence $(\zeta_n)_{n\geq 0}$,
and $(p_n)_{n\geq 0}$ a given sequence of positive numbers with the property that $0 < p_n < 1$ for all $n\geq 0$.
Then there exists a unique probability measure $\mathbb{P}$ on the measurable space $(\Omega,\mathscr{F})$ such that
\begin{equation}\label{eq-2-2}
\mathbb{P}\circ(\zeta_{n_1}, \zeta_{n_2}, \cdots, \zeta_{n_k})^{-1}\big\{(\epsilon_1, \epsilon_2, \cdots, \epsilon_k)\big\}
=\prod_{j=1}^k p_{n_j}^{\frac{1+\epsilon_j}{2}}(1-p_{n_j})^{\frac{1-\epsilon_j}{2}}
\end{equation}
for $n_j\in \mathbb{N}$, $\epsilon_j\in \{-1,1\}$\  $(1\leq j \leq k)$\ with $n_i\neq n_j$ when $i\neq j$
and $k\geq 1$. Thus one has a probability measure space $(\Omega, \mathscr{F}, \mathbb{P})$,
which is referred to as the Bernoulli space and complex-valued random variables on it are known as Bernoulli functionals.

Let $Z=(Z_n)_{n\geq 0}$ be the sequence of Bernoulli functionals generated by the sequence $(\zeta_n)_{n\geq 0}$, namely
\begin{equation}\label{eq-2-3}
   Z_n = \frac{\zeta_n + q_n - p_n}{2\sqrt{p_nq_n}},\quad n\geq0,
\end{equation}
where $q_n = 1-p_n$. Clearly $Z=(Z_n)_{n\geq 0}$ is a sequence of independent random variables on the
probability measure space $(\Omega, \mathscr{F}, \mathbb{P})$.
Let $\mathfrak{h}$ be the space of square integrable Bernoulli functionals, namely
\begin{equation}\label{eq-2-4}
  \mathfrak{h} = L^2(\Omega, \mathscr{F}, \mathbb{P}).
\end{equation}
We denote by
$\langle\cdot,\cdot\rangle_{\mathfrak{h}}$ the usual inner product of the space $\mathfrak{h}$, and by $\|\cdot\|_{\mathfrak{h}}$ the corresponding norm.
It is known \cite{privault,w-c-l} that $Z$ has the chaotic representation property, which implies that the system
$\mathfrak{Z} = \{Z_{\sigma}\mid \sigma\in \Gamma\}$  forms an orthonormal basis (ONB) for $\mathfrak{h}$, where
$Z_{\emptyset}=1$ and
\begin{equation}\label{eq-2-5}
    Z_{\sigma} = \prod_{j\in \sigma}Z_j,\quad \text{$\sigma \in \Gamma$, $\sigma \neq \emptyset$},
\end{equation}
where $\Gamma = \{\sigma \subset \mathbb{N} \mid \#\sigma<\infty\}$. The ONB $\mathfrak{Z} = \{Z_{\sigma}\mid \sigma\in \Gamma\}$ is known as the canonical ONB for $\mathfrak{h}$.
Clearly $\mathfrak{h}$ is separable and infinite-dimensional as a complex Hilbert space.

It can be shown that \cite{w-c-l}, for each $k\geq 0$, there exists a bounded operator $\partial_k$ on
$\mathfrak{h}$ such that
\begin{equation}\label{eq-2-6}
    \partial_k Z_{\sigma} = \mathbf{1}_{\sigma}(k)Z_{\sigma\setminus k},\quad \partial_k^{\ast} Z_{\sigma}
    = [1-\mathbf{1}_{\sigma}(k)]Z_{\sigma\cup k}\quad
    \sigma \in \Gamma,
\end{equation}
where $\partial_k^{\ast}$ denotes the adjoint of $\partial_k$, $\sigma\setminus k=\sigma\setminus \{k\}$, $\sigma\cup k=\sigma\cup \{k\}$
and $\mathbf{1}_{\sigma}(k)$ the indicator of $\sigma$ as a subset of $\mathbb{N}$.

The operators $\partial_k$ and $\partial_k^{\ast}$ are usually known as the annihilation and
creation operators acting on Bernoulli functionals, respectively.
And the family $\{\partial_k, \partial_k^{\ast}\}_{k \geq 0}$ is referred to as quantum Bernoulli noises (QBN).

A typical property of QBN is that they satisfy
the canonical anti-commutation relations (CAR) in equal-time \cite{w-c-l}.
More specifically, for $k$, $l\geq 0$, it holds true that
\begin{equation}\label{eq-2-7}
    \partial_k \partial_l = \partial_l\partial_k,\quad
    \partial_k^{\ast} \partial_l^{\ast} = \partial_l^{\ast}\partial_k^{\ast},\quad
    \partial_k^{\ast} \partial_l = \partial_l\partial_k^{\ast}\quad (k\neq l)
\end{equation}
and
\begin{equation}\label{eq-2-8}
   \partial_k\partial_k= \partial_k^{\ast}\partial_k^{\ast}=0,\quad
   \partial_k\partial_k^{\ast} + \partial_k^{\ast}\partial_k=I,
\end{equation}
where $I$ is the identity operator on $\mathfrak{h}$.

\section{QBN approach to quantum walk on hypercube}\label{sec-3}

In this section, we present our main work. We first give an alternative description of a hypercube,
and then define and characterize our coin operator systems. Based on these, we introduce our quantum walk model
on a general hypercube and examine its properties both from a perspective of pure mathematics and from a perspective of probability theory.
In particular, we show that our quantum walk can produce the uniform measure as its stationary measure.

\subsection{Alternative description of hypercube}\label{subsec-3-1}

In this subsection, we present an alternative description of a hypercube. To do so, we first recall some general
notions and notation about a graph.

A (simple) graph is a pair $(V,E)$, where $V$ is a nonempty set and known as the vertex set of the graph,
while $E$ is a subset of the set $\{e \mid e \subset V,\, \#e =2\}$ and known as the edge set of the graph, where $\#e$ means the cardinality of
$e$ as a subset of $V$. For vertices $v_1$, $v_2\in V$, if $\{v_1, v_2\}\subset E$, then it is said that
$v_1$ and $v_2$ are adjacent (there exists an edge linking $v_1$ and $v_2$), and written as $v_1 \sim v_2$.
A graph $(V,E)$ is said to be finite if its vertex set $V$ is a finite set.
For a vertex $v$ of a finite graph $(V,E)$, its degree is defined as $\mathrm{deg}(v)=\#\{v' \in V \mid v'\sim v\}$.
A finite graph $(V,E)$ is called a regular graph if $\mathrm{deg}(v)=d$ for all $v\in V$, where $d$ is some constant.
In that case, the constant $d$ is called the degree of the regular graph $(V,E)$.

We now recall the usual definition of a hypercube.
Let $n \geq 1$ be a positive integer and $V^{(n)}$ the $n$-fold Cartesian product of the set $\{0,1\}$, namely
\begin{equation}\label{eq}
  V^{(n)}=\big\{\, \mathrm{x}=(x_1,x_2,\cdots, x_n) \mid x_k \in \{0,1\},\, 1\leq k \leq n\,\big\}.
\end{equation}
The $n$-dimensional hypercube is then the graph $(V^{(n)}, E^{(n)})$ with $V^{(n)}$ being the vertex set and $E^{(n)}$ being the edge set,
where $E^{(n)}$ is given by
\begin{equation}\label{eq}
  E^{(n)}=\big\{\, \{\mathrm{x},\mathrm{y}\} \mid \mathrm{x},\mathrm{y} \in V^{(n)},\, |\mathrm{x}-\mathrm{y}|=1\, \big\}
\end{equation}
with $|\mathrm{x}-\mathrm{y}|$ denoting the Hamming distance between $\mathrm{x}$ and $\mathrm{y}$, which is given by
\begin{equation}\label{eq}
|\mathrm{x}-\mathrm{y}| =\sum_{k=1}^n |x_k-y_k|,\quad \mathrm{x}=(x_1,x_2,\cdots, x_n),\, \mathrm{y}=(y_1,y_2,\cdots, y_n).
\end{equation}
It is well known that the $n$-dimensional hypercube $(V^{(n)}, E^{(n)})$ is a regular graph and its degree is exactly $n$.
Moreover, it has $2^n$ vertices and $n\times 2^{n-1}$ edges (see, e.g. \cite{venegas} for more information  about a hypercube).

We next give an alternative description of a hypercube. For a nonnegative integer $n\geq 0$, we write $\mathbb{N}_n= \{0,1,2,\cdots, n\}$
and denote by $\Gamma_n$ its the power set, namely
\begin{equation}\label{eq}
  \Gamma_n=\{\sigma \mid \sigma\subset \mathbb{N}_n\}.
\end{equation}
For example, $\Gamma_1=\big\{\,\emptyset,\, \{0\},\, \{1\},\, \{0,1\}\,\big\}$.
Recall that, for $\sigma\in \Gamma_n$ and $k\in \mathbb{N}_n$, we use $\sigma\setminus k$ to mean $\sigma\setminus \{k\}$ for brevity.
Similarly, we use $\sigma\cup k$, $\sigma\cap k$, etc.

\begin{definition}
For nonnegative integer $n\geq 0$, we denote by $(\Gamma_n, \mathfrak{E}_n)$ the graph with $\Gamma_n$ being the vertex set and
$\mathfrak{E}_n$ being the edge set, where $\mathfrak{E}_n$ given by
\begin{equation}\label{eq}
 \mathfrak{E}_n=\big\{\, \{\sigma, \tau\} \mid \sigma, \tau\in \Gamma_n,\, \#(\sigma\bigtriangleup \tau)=1\,\big\},
\end{equation}
where $\sigma\bigtriangleup \tau = (\sigma\setminus \tau) \cup (\tau\setminus\sigma)$ and $\#(\sigma\bigtriangleup \tau)$ means the
cardinality of the set $\sigma\bigtriangleup \tau$.
\end{definition}

The following proposition shows that the graph $(\Gamma_n, \mathfrak{E}_n)$ actually belongs to the category of hypercubes.

\begin{proposition}
Let $n\geq 0$ be a nonnegative integer. Then the graph $(\Gamma_n, \mathfrak{E}_n)$ is isomorphic to the $(n+1)$-dimensional hypercube $(V^{(n+1)}, E^{(n+1)})$.
\end{proposition}

\begin{proof}
Define a mapping $J\colon \Gamma_n\rightarrow V^{(n+1)}$ in the manner as follows
\begin{equation*}
  J(\sigma) = \big(\mathbf{1}_{\sigma}(0), \mathbf{1}_{\sigma}(1),\cdots , \mathbf{1}_{\sigma}(n)\big),\quad \sigma \in  \Gamma_n,
\end{equation*}
where $\mathbf{1}_{\sigma}(x)$ denotes the indicator of $\sigma$ as a subset of $\mathbb{N}_n$. It is easy to see that
$J$ is a bijection from $\Gamma_n$ to $V^{(n+1)}$. Let $\sigma$, $\tau\in \Gamma_n$ be such that $\sigma\sim \tau$.
Then $\#(\sigma\bigtriangleup \tau)=1$, which implies that $\sigma\subset \tau$ with $\#(\tau\setminus \sigma)=1$
or $\tau \subset \sigma$ with $\#(\sigma\setminus \tau)=1$. In the case of $\sigma\subset \tau$ with $\#(\tau\setminus \sigma)=1$,
we have
\begin{equation*}
  |J(\sigma)-J(\tau)| = \sum_{k=0}^{n-1}|\mathbf{1}_{\sigma}(k)-\mathbf{1}_{\tau}(k)|
  = \sum_{k\in \tau\setminus\sigma}\mathbf{1}_{\tau}(k)=1,
\end{equation*}
which means $J(\sigma)\sim J(\tau)$. Similarly, we can also get $J(\sigma)\sim J(\tau)$ in the case of
$\tau\subset \sigma$ with $\#(\sigma\setminus \tau)=1$.
\end{proof}

Consider the graph $(\Gamma_n, \mathfrak{E}_n)$, where $n\geq 0$.
For vertices $\sigma$ and $\tau\in \Gamma_n$, as usual we use $\sigma\sim \tau$ to mean that $\sigma$ and $\tau$ are adjacent,
namely $\{\sigma,\tau\} \in \mathfrak{E}_n$.

\begin{proposition}\label{prop-3-2}
Let $\sigma$, $\tau\in \Gamma_n$ be vertices in the graph $(\Gamma_n, \mathfrak{E}_n)$. Then, $\sigma\sim \tau$ if and only if
there exists a unique $k\in \mathbb{N}_n$ such that $k\in \sigma$ with $\sigma\setminus k = \tau$ or $k\notin \sigma$ with $\sigma \cup k = \tau$.
\end{proposition}

\begin{proof}
Let $\sigma\sim \tau$. Then $\#(\sigma\bigtriangleup \tau)=1$, which implies that $\#(\sigma\setminus \tau)=1$ with $\tau\subset \sigma$ or
$\#(\tau\setminus \sigma)=1$ with $\sigma\subset \tau$, which implies that there exists a unique $k\in \mathbb{N}_n$
such that $k\in \sigma$ with $\sigma\setminus k = \tau$ or $k\notin \sigma$ with $\sigma \cup k = \tau$.

Conversely, if there exists a unique $k\in \mathbb{N}_n$ such that $k\in \sigma$ with $\sigma\setminus k = \tau$ or $k\notin \sigma$ with $\sigma \cup k = \tau$,
then $\sigma\bigtriangleup \tau=\{k\}$, hence $\#(\sigma\bigtriangleup \tau)=1$, which means $\{\sigma,\tau\}\in \mathfrak{E}_n$,
namely $\sigma\sim \tau$.
\end{proof}

\begin{corollary}\label{col-2-3}
Let $\sigma\in \Gamma_n$ be a vertex in the graph  $(\Gamma_n, \mathfrak{E}_n)$. Then, for $k\in \mathbb{N}_n$, one has: $\sigma\sim (\sigma\setminus k)$ when $k\in \sigma$;\
$\sigma\sim (\sigma\cup k)$ when $k\notin \sigma$.
\end{corollary}

\subsection{Coin operator system}\label{subsec-3-2}

To define a quantum walk model, one needs some operators on a Hilbert space to describe the walker's internal degrees of freedom.
Such operators are usually known as coin operators, while the space they act on is referred to as the coin space.
This subsection defines a notion of coin operator system and examine properties of such system from a perspective of pure mathematics.

In this subsection, we assume that $n\geq 0$ is a given nonnegative integer and $\mathcal{K}$ is a Hilbert space with $\dim \mathcal{K}\geq n+1$
(namely the dimension of $\mathcal{K}$ is not less than $n+1$).
We denote by $\langle \cdot,\cdot\rangle_{\mathcal{K}}$ and $\|\cdot\|_{\mathcal{K}}$ the inner product and norm in $\mathcal{K}$, respectively.

\begin{definition}\label{def-coin-operator}
A system $\mathfrak{C}=\{C_k \mid 0\leq k \leq n\}$ of bounded operators on $\mathcal{K}$ is called a coin operator system
if the sum $\sum_{k=0}^n C_k$ is a unitary operator on $\mathcal{K}$ and
\begin{equation}\label{eq}
C^*_jC_k= C_j C_k^*=0,\quad  j\neq k,\, 0\leq j,\, k \leq n,
\end{equation}
where $C_k^*$ denotes the adjoint of $C_k$.
\end{definition}

Let $\mathfrak{C}=\{C_k \mid 0\leq k \leq n\}$ be a coin operator system on $\mathcal{K}$. Then, it follows immediately that
\begin{equation}\label{eq}
  \sum_{k=0}^n C_k^*C_k = \sum_{k=0}^n C_kC_k^* =I,
\end{equation}
where $I$ denotes the the identity operator on $\mathcal{K}$. Based on these equalities, one can further know that
both $C_k^*C_k$ and $C_kC_k^*$ are projection operators on $\mathcal{K}$ for each $k$ with $0\leq k \leq n$.
To characterize a coin operator system, let us first introduce a notion as follows.

A system $\{P_k \mid 0\leq k \leq n\}$ of projection operators on the space $\mathcal{K}$ is called a resolution of the identity if
$\sum_{k=0}^n P_k = I$ and $P_jP_k =0$, $j\neq k$, $0\leq j$, $k\leq n$.

In terms of a resolution of the identity as well as a unitary operator, the next theorem provides a necessary and sufficient condition
for a system of bounded operators to be a coin operator system.

\begin{theorem}\label{coin-characterization}
Let\ $\mathfrak{C}=\{ C_k \mid 0\leq k\leq n\}$ be a system of bounded operators on $\mathcal{K}$. Then the following statements are equivalent:
\begin{enumerate}
  \item[(1)] The system\ $\mathfrak{C}=\{ C_k \mid 0\leq k\leq n\}$ is a coin operator system on $\mathcal{K}$.
  \item[(2)] There exist a unitary operator $U$ and a resolution of the identity $\{P_k \mid 0\leq k \leq n\}$ on $\mathcal{K}$ such that
   $C_k = P_kU$,\ \ $0\leq k\leq n$.
\end{enumerate}
\end{theorem}

\begin{proof}
``(1)\ $\Rightarrow$\ (2)''. Let $U=\sum_{k=0}^n C_k$. Then, by the definition, $U$ is a unitary operator.
Now consider the system $\{P_k \mid 0\leq k \leq n\}$ of projection operators, where $P_k = C_kC_k^*$. Clearly,
it is a resolution of the identity. For $0\leq k\leq n$, a straightforward calculation yields
\begin{equation*}
  P_kU = C_kC_k^*\sum_{j=0}^n C_j = C_kC_k^*C_k  = C_k\Big(I- \sum_{j=0,j\neq k}^n C_j^*C_j\Big) = C_k.
\end{equation*}
``(2)\ $\Rightarrow$\ (1)''. It is easy to see that $\sum_{k=0}^n C_k = U$, which means that $\sum_{k=0}^n C_k$ is a unitary operator.
For $0\leq j$, $k\leq n$ with $j\neq k$, using $P_j P_k=0$ gives
\begin{equation*}
  C_j^*C_k = (P_jU)^*P_kU = U^*P_j P_kU =0,\quad C_jC_k^* = P_jU(P_kU)^* = P_j U U^*P_k =P_jP_k=0.
\end{equation*}
Thus $\mathfrak{C}=\{ C_k \mid 0\leq k\leq n\}$ is a coin operator system.
\end{proof}

Recall that, for $\tau\in \Gamma_n$, $\mathbf{1}_{\tau}(k)$ stands for the indicator of the set $\tau$,
which allows us to define a function $\varepsilon_{\tau}(k)$ on $\mathbb{N}_n$ as
\begin{equation}\label{eq}
\varepsilon_{\tau}(k)= 2\times\mathbf{1}_{\tau}(k)-1,\quad k\in \mathbb{N}_n.
\end{equation}
Clearly, the function $\varepsilon_{\tau}(k)$ takes values in $\{-1,\, 1\}$.

\begin{definition}\label{def-weighted-sum-of-coin}
For $\tau\in \Gamma_n$ and a coin operator system $\mathfrak{C}=\{C_k \mid 0\leq k \leq n\}$ on $\mathcal{K}$, we define
\begin{equation}\label{eq}
  U_{\tau}^{(\mathfrak{C})}=\sum_{k=0}^n \varepsilon_{\tau}(k)C_k
\end{equation}
and call it the $\varepsilon_{\tau}$-weighted sum of the coin operator system $\mathfrak{C}$.
\end{definition}

\begin{proposition}\label{prop-weighted-sum-of-coin}
Let $\mathfrak{C}=\{C_k \mid 0\leq k \leq n\}$ be a coin operator system on $\mathcal{K}$. Then
$U_{\tau}^{(\mathfrak{C})}$ is a unitary operator on $\mathcal{K}$ for each $\tau\in \Gamma_n$.
\end{proposition}

\begin{proof}
Let $\tau\in \Gamma_n$ be given. Using properties of the coin operator system $\mathfrak{C}$, we have
\begin{equation*}
  \big(U_{\tau}^{(\mathfrak{C})}\big)^*U_{\tau}^{(\mathfrak{C})}
= \Big(\sum_{k=0}^n \varepsilon_{\tau}(k)C_k^*\Big)\Big(\sum_{k=0}^n \varepsilon_{\tau}(k)C_k\Big)
= \sum_{k=0}^n \big(\varepsilon_{\tau}(k)\big)^2C_k^*C_k
= \sum_{k=0}^n C_k^*C_k
= I,
\end{equation*}
where $I$ stands for the identity operator on $\mathcal{K}$. Similarly, we have $U_{\tau}^{(\mathfrak{C})}\big(U_{\tau}^{(\mathfrak{C})}\big)^*=I$.
Thus $U_{\tau}^{(\mathfrak{C})}$ is a unitary operator.
\end{proof}

\subsection{Definition of the quantum walk model}\label{subsec-3-3}

In this subsection, we define our quantum walk model and examine its fundamental properties.
Throughout this subsection, we assume that $n\geq 0$ is a fixed nonnegative integer and
$\mathcal{K}$ is a finite-dimensional Hilbert space with $d_{\mathcal{K}}\equiv\dim \mathcal{K}\geq n+1$.
Additionally, we fix an orthonormal basis $\{e_j \mid 0\leq j\leq d_{\mathcal{K}}-1\}$ for $\mathcal{K}$.
We will take $\mathcal{K}$ as the coin space for our quantum walk model.

Recall that the space $\mathfrak{h}$ of square integrable Bernoulli functionals has an orthonormal basis
of form $\{Z_{\sigma} \mid \sigma \in \Gamma\}$, which is known as its canonical ONB.
We denote by $\mathfrak{h}_n$ the subspace of $\mathfrak{h}$ spanned by $\{Z_{\sigma} \mid \sigma \in \Gamma_n\}$, namely
\begin{equation}\label{eq}
  \mathfrak{h}_n = \mathrm{span}\{Z_{\sigma} \mid \sigma \in \Gamma_n\}.
\end{equation}
Note that $\Gamma_n\subset \Gamma$ and $\#(\Gamma_n)=2^{n+1}$, which implies that $\mathfrak{h}_n$
is a $2^{n+1}$-dimensional subspace of $\mathfrak{h}$, hence a closed subspace.
In other words, $\mathfrak{h}_n$ itself is a Hilbert space with the inner product $\langle\cdot,\cdot\rangle_{\mathfrak{h}}$.
It can be shown that, for all $k\in \mathbb{N}_n$, both the annihilation operator $\partial_k$ and the creation operator $\partial_k^*$
leave $\mathfrak{h}_n$ invariant.
This means that, for all $k\in \mathbb{N}_n$,  $\partial_k$ and $\partial_k^*$ can be viewed as the annihilation and creation operators on $\mathfrak{h}_n$,
respectively.

\begin{theorem}
Consider the graph $(\Gamma_n,\mathfrak{E}_n)$. Let $\sigma$, $\tau\in \Gamma_n$ be its vertices.
Then $\sigma\sim \tau$ if and only if there exists a unique $k\in \mathbb{N}_n$ such that
\begin{equation}\label{eq}
  (\partial_k^*+\partial_k)Z_{\sigma} = Z_{\tau}.
\end{equation}
\end{theorem}

\begin{proof}
Let $\tau\in \Gamma_n$. Then, by Proposition~\ref{prop-3-2}, there exists a unique $k\in \mathbb{N}_n$ such that $k\in \sigma$ with $\tau=\sigma\setminus k$
or $k\notin \sigma$ with $\tau=\sigma\cup k$, which implies that
\begin{equation*}
Z_{\tau}= 1_{\sigma}(k) Z_{\sigma\setminus k} + (1-1_{\sigma}(k))Z_{\sigma\cup k} = (1-1_{\sigma}(k))Z_{\sigma\cup k} +  1_{\sigma}(k) Z_{\sigma\setminus k}.
\end{equation*}
On the other hand, it follows from properties of $\partial_k$ and $\partial_k^*$ (see (\ref{eq-2-6}) for details) that
\begin{equation*}
(\partial_k^* + \partial_k)Z_{\sigma} = (1-1_{\sigma}(k))Z_{\sigma\cup k} + 1_{\sigma}(k) Z_{\sigma\setminus k}.
\end{equation*}
Thus $(\partial_k^* + \partial_k)Z_{\sigma}=Z_{\tau}$. Now suppose that $(\partial_k^* + \partial_k)Z_{\sigma}=Z_{\tau}$. Then, we have
\begin{equation*}
  Z_{\tau}= (1-1_{\sigma}(k))Z_{\sigma\cup k} + 1_{\sigma}(k) Z_{\sigma\setminus k},
\end{equation*}
which implies that $k\in \sigma$ with $\tau=\sigma\setminus k$ or $k\notin \sigma$ with $\tau=\sigma\cup k$,
which together with Proposition~\ref{prop-3-2} implies $\sigma\sim \tau$ .
\end{proof}

\begin{remark}
Let $\sigma$ be a vertex in the graph $(\Gamma_n,\mathfrak{E}_n)$ and $k\in \mathbb{N}_n$. Then, using properties of $\partial_k$ and $\partial_k^*$, we have
\begin{equation*}
  (\partial_k^*+\partial_k)Z_{\sigma}
= \left\{
    \begin{array}{ll}
      Z_{\sigma\setminus k}, & \hbox{$k\in \sigma$;}\\
      Z_{\sigma\cup k}, & \hbox{$k\notin \sigma$.}
    \end{array}
  \right.
\end{equation*}
On the hand, by Corollary~\ref{col-2-3}, $\sigma\sim (\sigma\setminus k)$ when $k\in \sigma$; $\sigma\sim (\sigma\cup k)$ when $k\notin \sigma$.
Therefore, the operator $(\partial_k^*+\partial_k)$  on $\mathfrak{h}_n$ behaves actually as a shift operator.
\end{remark}

As mentioned above, the Hilbert space $\mathcal{K}$ will serve as the coin space of our quantum walk model.
In what follows, we denote by $\langle\cdot,\cdot\rangle$ and $\|\cdot\|$ the inner product and norm in the tensor space $\mathfrak{h}_n\otimes \mathcal{K}$.
We mention that $\{Z_{\sigma}\otimes e_j \mid \sigma\in \Gamma_n,\, 0\leq j\leq d_{\mathcal{K}}-1\}$ is an orthonormal basis for $\mathfrak{h}_n\otimes \mathcal{K}$.

\begin{theorem}
Let $\mathfrak{C}=\{C_k \mid 0\leq  k \leq n\}$ be a coin operator system on the coin space $\mathcal{K}$ and write
\begin{equation}\label{evolution-operator}
  \mathsf{W}\!_{\mathfrak{C}}= \sum_{k=0}^n(\partial_k^*+\partial_k)\otimes C_k.
\end{equation}
Then $\mathsf{W}\!_{\mathfrak{C}}$ is a unitary operator on $\mathfrak{h}_n\otimes \mathcal{K}$.
\end{theorem}

\begin{proof}
It follows from (\ref{eq-2-7}) and (\ref{eq-2-8}) that $(\partial_k^*+\partial_k)^2=I_{\mathfrak{h}_n}$, $0\leq k \leq n$,
where $I_{\mathfrak{h}_n}$ denotes the identity operator on $\mathfrak{h}_n$. Thus, by Definition~\ref{def-coin-operator}
and properties of coin operator systems, we have
\begin{equation*}
\begin{split}
  {\mathsf{W}\!_{\mathfrak{C}}}^*\mathsf{W}\!_{\mathfrak{C}}
  &= \Big(\sum_{k=0}^n(\partial_k^*+\partial_k)\otimes C_k^*\Big)\Big(\sum_{k=0}^n(\partial_k^*+\partial_k)\otimes C_k\Big)\\
  &= \sum_{k=0}^n(\partial_k^*+\partial_k)^2\otimes C_k^*C_k\\
  &= \sum_{k=0}^nI_{\mathfrak{h}_n}\otimes C_k^*C_k\\
  &= I_{\mathfrak{h}_n}\otimes I_{\mathcal{K}}\\
  &=I,
\end{split}
\end{equation*}
where $I_{\mathcal{K}}$ means the identity operator on $\mathcal{K}$. Similarly, we have $\mathsf{W}\!_{\mathfrak{C}}{\mathsf{W}\!_{\mathfrak{C}}}^*=I$.
\end{proof}

In the following, unless otherwise specified, we always assume that $\mathfrak{C}=\{C_k \mid  0\leq k \leq n\}$ is a fixed coin operator system on the space $\mathcal{K}$.
We call the operator $\mathsf{W}\!_{\mathfrak{C}}$ defined by (\ref{evolution-operator})
the unitary operator on $\mathfrak{h}_n\otimes \mathcal{K}$ generated by the coin operator system $\mathfrak{C}$.
The next definition describes our quantum walk model on the graph $(\Gamma_n, \mathfrak{E}_n)$.

\begin{definition}\label{def-model}
The quantum walk on the graph $(\Gamma_n, \mathfrak{E}_n)$ with $\mathsf{W}\!_{\mathfrak{C}}$ being the evolution operator is the discrete-time quantum walk
that admits the following features:
\begin{itemize}
  \item The walk takes $\mathfrak{h}_n\otimes \mathcal{K}$ as its state space, where $\mathfrak{h}_n$
       describes the position information of the walk, while $\mathcal{K}$ describes its internal degrees of freedom;
  \item The states of the walk are represented by unit vectors in $\mathfrak{h}_n\otimes \mathcal{K}$ and the time evolution of the walk is governed by the equation
        \begin{equation}\label{eq}
          \Phi_{t+1} = \mathsf{W}\!_{\mathfrak{C}}\Phi_t,\quad t\geq 0,
        \end{equation}
  where $\Phi_t$ denotes the state of the walk at time $t\geq 0$, especially $\Phi_0$ denotes the initial state of the walk;
  \item The probability $P_t(\sigma\,|\,\Phi_0)$ of finding the walker on vertex $\sigma\in \Gamma_n$ at time $t\geq 0$ is given by
      \begin{equation}\label{eq}
        P_t(\sigma\,|\,\Phi_0) = \sum_{j=0}^{d_{\mathcal{K}}-1} |\langle Z_{\sigma}\otimes e_j, \Phi_t\rangle|^2,
      \end{equation}
 where $\{e_j\mid 0\leq j\leq d_{\mathcal{K}}-1\}$ is the orthonormal basis for the coin space $\mathcal{K}$.
\end{itemize}

\end{definition}

Conventionally, $\mathfrak{h}_n$ and $\mathcal{K}$ are known as the position space and coin space of the walk, respectively,
while the function $\sigma\mapsto P_t(\sigma\,|\,\Phi_0)$ on $\Gamma_n$ is called the probability distribution of the walk at time $t\geq 0$,
which usually depends on the initial state $\Phi_0$ and the evolution operator $\mathsf{W}\!_{\mathfrak{C}}$.

\begin{remark}
In what follows, we simply use $\mathsf{W}\!_{\mathfrak{C}}$ to indicate the quantum walk on the graph $(\Gamma_n, \mathfrak{E}_n)$
with $\mathsf{W}\!_{\mathfrak{C}}$ being the evolution operator. In other words, when we say the walk $\mathsf{W}\!_{\mathfrak{C}}$, we just mean
the quantum walk on the graph $(\Gamma_n, \mathfrak{E}_n)$ with $\mathsf{W}\!_{\mathfrak{C}}$ being the evolution operator.
\end{remark}

\subsection{Probability distribution of the quantum walk}\label{subsec-3-4}

In the present subsection, we examine properties of the walk $\mathsf{W}\!_{\mathfrak{C}}$ from a perspective of probability distribution.
We continue to use the assumptions made in the previous subsection.
Unless otherwise specified, we always use $\Phi_0$ to mean the initial state of the walk $\mathsf{W}\!_{\mathfrak{C}}$.

Recall that, for $\sigma\in \Gamma_n$, the function $\varepsilon_{\sigma}(\cdot)$ on $\mathbb{N}_n$ is defined by
$\varepsilon_{\sigma}(k)= 2\times\mathbf{1}_{\sigma}(k)-1$,
where $\mathbf{1}_{\sigma}(k)$ is the indicator of $\sigma$. With these functions, we introduce the following vectors in
the position space $\mathfrak{h}_n$
\begin{equation}\label{eq-3-12}
  \widehat{Z}_{\sigma}= \frac{1}{\sqrt{2^{n+1}}}\sum_{\tau\in \Gamma_n} \Big(\prod_{k\in \tau}\varepsilon_{\sigma}(k)\Big)Z_{\tau},\quad \sigma\in \Gamma_n,
\end{equation}
where $\prod_{k\in \tau}\varepsilon_{\sigma}(k)=1$ if $\tau=\emptyset$. For $\sigma$, $\tau\in \Gamma_n$, one can easily obtain the next useful formula
\begin{equation}\label{eq}
  \langle Z_{\sigma}, \widehat{Z}_{\tau}\rangle_{\mathfrak{h}} = \frac{1}{\sqrt{2^{n+1}}}\prod_{k\in \sigma}\varepsilon_{\tau}(k)
  = \frac{(-1)^{\#(\sigma\setminus \tau)}}{\sqrt{2^{n+1}}},
\end{equation}
where $\#(\sigma\setminus \tau)$ denotes the cardinality of the set $\sigma\setminus \tau$.
Additionally, we also introduce the following operators on the position space $\mathfrak{h}_n$
\begin{equation}\label{def-of-A-sigma}
  A_{\sigma} = \prod_{k=0}^n \big(I + \varepsilon_{\sigma}(k)(\partial_k^* + \partial_k)\big),\quad \sigma\in \Gamma_n,
\end{equation}
where $I$ is the identity operator on $\mathfrak{h}_n$.

\begin{proposition}\label{prop-3-7}
Let $\sigma\in \Gamma_n$ be given. Then $\widehat{Z}_{\sigma} =\frac{1}{\sqrt{2^{n+1}}} A_{\sigma}Z_{\emptyset}$, where $Z_{\emptyset}$
is the basis vector in the canonical ONB for $\mathfrak{h}_n$.
\end{proposition}

\begin{proof}
Consider the operator system $\mathfrak{A}_{\sigma}=\big\{ I + \varepsilon_{\sigma}(k)(\partial_k^* + \partial_k) \mid 0\leq k\leq n\big\}$
associated with $\sigma$. According to (\ref{eq-2-7}), any two operators in $\mathfrak{A}_{\sigma}$ are
commutative. Thus, by a straightforward calculation, we can get
\begin{equation}\label{expansion-of-A-sigma}
  A_{\sigma}
= \sum_{\tau\in \Gamma_n} \Big(\prod_{k\in \tau}\varepsilon_{\sigma}(k)\Big) \Big(\prod_{k\in \tau}(\partial_k^* + \partial_k)\Big),
\end{equation}
where $\prod_{k\in \tau}(\partial_k^* + \partial_k)=I$ when $\tau=\emptyset$.
On the other hand, for each $\tau\in \Gamma_n$, by using (\ref{eq-2-6}) and the induction method we have
\begin{equation*}
  \Big(\prod_{k\in \tau}(\partial_k^* + \partial_k)\Big)Z_{\emptyset} = Z_{\tau},
\end{equation*}
which together with (\ref{expansion-of-A-sigma}) implies $A_{\sigma}Z_{\emptyset}=\sqrt{2^{n+1}}\,\widehat{Z}_{\sigma}$, equivalently
$\widehat{Z}_{\sigma}= \frac{1}{\sqrt{2^{n+1}}}A_{\sigma}Z_{\emptyset}$.
\end{proof}

\begin{proposition}\label{prop-3-8}
Let $\sigma\in \Gamma_n$ be given. Then, for all $k\in \mathbb{N}_n$,  the following formula
holds true
\begin{equation}\label{absortion-A-sigma}
 (\partial_k^* + \partial_k)A_{\sigma} = \varepsilon_{\sigma}(k)A_{\sigma}.
\end{equation}
\end{proposition}

\begin{proof}
Let $k\in \mathbb{N}_n$ be given. Then, by CAR in equal time (\ref{eq-2-8}), we find
\begin{equation*}
  (\partial_k^* + \partial_k)\big(I + \varepsilon_{\sigma}(k)(\partial_k^* + \partial_k)\big)
= \varepsilon_{\sigma}(k)\big(I + \varepsilon_{\sigma}(k)(\partial_k^* + \partial_k)\big),
\end{equation*}
which, together with the commutativity of the operator system $\mathfrak{A}_{\sigma}$, gives
\begin{equation*}
\begin{split}
(\partial_k^* + \partial_k)A_{\sigma}
 & = (\partial_k^* + \partial_k)\big(I + \varepsilon_{\sigma}(k)(\partial_k^* + \partial_k)\big)
     \prod_{j=0,j\neq k}^n \big(I + \varepsilon_{\sigma}(j)(\partial_j^* + \partial_j)\big)\\
 & = \varepsilon_{\sigma}(k)\big(I + \varepsilon_{\sigma}(k)(\partial_k^* + \partial_k)\big)
     \prod_{j=0,j\neq k}^n \big(I + \varepsilon_{\sigma}(j)(\partial_j^* + \partial_j)\big)\\
 & = \varepsilon_{\sigma}(k)\prod_{j=0}^n \big(I + \varepsilon_{\sigma}(j)(\partial_j^* + \partial_j)\big)\\
 & = \varepsilon_{\sigma}(k)A_{\sigma}.
\end{split}
\end{equation*}
This completes the proof.
\end{proof}

\begin{theorem}\label{thr-new-ONB}
The vector system $\big\{\widehat{Z}_{\sigma} \mid \sigma \in \Gamma_n\big\}$ forms an orthonormal basis for the position space $\mathfrak{h}_n$.
Moreover, it holds true that
\begin{equation}\label{eq}
  (\partial_k^*+\partial_k)\widehat{Z}_{\sigma}= \varepsilon_{\sigma}(k)\widehat{Z}_{\sigma},\quad k\in \mathbb{N}_n,\, \sigma\in \Gamma_n.
\end{equation}
\end{theorem}

\begin{proof}
Let $\sigma$, $\gamma\in \Gamma_n$ be such that $\sigma\neq \gamma$. Then, there exists $k\in \mathbb{N}_n$ such that
$\varepsilon_{\sigma}(k) \neq \varepsilon_{\gamma}(k)$, which implies that $\varepsilon_{\sigma}(k)\varepsilon_{\gamma}(k)=-1$
and $\varepsilon_{\sigma}(k)+\varepsilon_{\gamma}(k)=0$. Thus, we have
\begin{equation*}
  \big(I + \varepsilon_{\sigma}(k)(\partial_k^* + \partial_k)\big)\big(I + \varepsilon_{\gamma}(k)(\partial_k^* + \partial_k)\big)
 = (1+ \varepsilon_{\sigma}(k)\varepsilon_{\gamma}(k))I + (\varepsilon_{\sigma}(k)+\varepsilon_{\gamma}(k))(\partial_k^* + \partial_k)=0,
\end{equation*}
which, together with (\ref{def-of-A-sigma}) as well as the commutativity of both $\mathfrak{A}_{\sigma}$ and $\mathfrak{A}_{\gamma}$, yields
\begin{equation*}
\begin{split}
  A_{\sigma} A_{\gamma}
   & =  \big(I + \varepsilon_{\sigma}(k)(\partial_k^* + \partial_k)\big)\big(I + \varepsilon_{\gamma}(k)(\partial_k^* + \partial_k)\big)\\
   &\quad \times \prod_{j=0,j\neq k}^n \big(I + \varepsilon_{\sigma}(j)(\partial_j^* + \partial_j)\big)
         \prod_{j=0,j\neq k}^n \big(I + \varepsilon_{\gamma}(j)(\partial_j^* + \partial_j)\big)\\
   & =0,
\end{split}
\end{equation*}
which, together with Proposition~\ref{prop-3-7} and the self-adjoint property of $A_{\sigma}$, further gives
\begin{equation*}
  \langle \widehat{Z}_{\sigma},\widehat{Z}_{\gamma}\rangle_{\mathfrak{h}}
   = \frac{1}{2^{n+1}}\langle  A_{\sigma}Z_{\emptyset}, A_{\gamma}Z_{\emptyset}\rangle_{\mathfrak{h}}
   = \frac{1}{2^{n+1}}\langle  Z_{\emptyset}, A_{\sigma}A_{\gamma}Z_{\emptyset}\rangle_{\mathfrak{h}}
   = 0.
\end{equation*}
It follows directly from (\ref{eq-3-12}) that
\begin{equation*}
  \langle\widehat{Z}_{\sigma}, \widehat{Z}_{\sigma}\rangle_{\mathfrak{h}}
   = \|\widehat{Z}_{\sigma}\|_{\mathfrak{h}}^2
   = \frac{1}{2^{n+1}}\sum_{\tau\in \Gamma_n} \Big\|\Big(\prod_{k\in \tau}\varepsilon_{\sigma}(k)\Big)Z_{\tau}\Big\|_{\mathfrak{h}}^2
   =  \frac{1}{2^{n+1}}\sum_{\tau\in \Gamma_n}1
   = 1.
\end{equation*}
Therefore, the system $\big\{\widehat{Z}_{\sigma} \mid \sigma \in \Gamma_n\big\}$ is an orthonormal system in $\mathfrak{h}_n$.
This, together with the fact that $\#\big\{\widehat{Z}_{\sigma} \mid \sigma \in \Gamma_n\big\} =\dim\mathfrak{h}_n=2^{n+1}$, means that
$\big\{\widehat{Z}_{\sigma} \mid \sigma \in \Gamma_n\big\}$ actually forms an orthonormal basis for $\mathfrak{h}_n$.
Finally, for $k\in \mathbb{N}_n$ and $\sigma\in \Gamma_n$, using Proposition~\ref{prop-3-7} and Proposition~\ref{prop-3-8} leads to
\begin{equation*}
  (\partial_k^*+\partial_k)\widehat{Z}_{\sigma}
 =\frac{1}{\sqrt{2^{n+1}}}(\partial_k^*+\partial_k) A_{\sigma}Z_{\emptyset}
= \frac{1}{\sqrt{2^{n+1}}}\varepsilon_{\sigma}(k)A_{\sigma}Z_{\emptyset}
= \varepsilon_{\sigma}(k)\widehat{Z}_{\sigma}.
\end{equation*}
In summary, the theorem is true.
\end{proof}

As an immediate consequence of Theorem~\ref{thr-new-ONB}, we have the following corollary, which shows that all
the shift operators $\partial_k^* +\partial_k$,  $k\in \mathbb{N}_n$, have a common fixed point in the unit sphere of $\mathfrak{h}_n$.

\begin{corollary}
Write $\xi^\star= \widehat{Z}_{\mathbb{N}_n}$. Then\ $(\partial_k^* +\partial_k) \xi^\star=\xi^\star$,\ \ $\forall\, k\in \mathbb{N}_n$.
\end{corollary}

\begin{proof}
Let $k\in \mathbb{N}_n$ be given. Then $\varepsilon_{\mathbb{N}_n}(k)= 2\times 1_{\mathbb{N}_n}(k)-1=1$, which together with Theorem~\ref{thr-new-ONB} gives
\begin{equation*}
(\partial_k^* +\partial_k)\xi^\star
=(\partial_k^* +\partial_k)\widehat{Z}_{\mathbb{N}_n}
=\varepsilon_{\mathbb{N}_n}(k)\widehat{Z}_{\mathbb{N}_n}
= \widehat{Z}_{\mathbb{N}_n}
= \xi^\star,
\end{equation*}
which is the desired.
\end{proof}

\begin{remark}\label{rem-3-3}
According to Theorem~\ref{thr-new-ONB}, the vector system $\{\widehat{Z}_{\sigma}\otimes e_j \mid \sigma \in \Gamma_n,\, 0\leq j\leq d_{\mathcal{K}}-1\}$
forms an orthonormal basis for the tensor space $\mathfrak{h}_n\otimes \mathcal{K}$. Thus, each $\Phi\in \mathfrak{h}_n\otimes \mathcal{K}$ has an expansion
of the following form
\begin{equation}
  \Phi = \sum_{\sigma\in \Gamma_n}\widehat{Z}_{\sigma}\otimes u_{\sigma},
\end{equation}
where $u_{\sigma}= \sum_{j=0}^{d_{\mathcal{K}}-1}\langle \widehat{Z}_{\sigma}\otimes e_j, \Phi\rangle e_j$.
\end{remark}

Recall that, for $\tau\in \Gamma_n$, the $\varepsilon_{\tau}$-weighted sum $U_{\tau}^{(\mathfrak{C})}$ of the coin operator system $\mathfrak{C}$
is a unitary operator on the coin space $\mathcal{K}$ (see Definition~\ref{def-weighted-sum-of-coin} and Proposition~\ref{prop-weighted-sum-of-coin}).
The next result unveils a link between the action of the evolution operator $\mathsf{W}\!_{\mathfrak{C}}$ and
that of $U_{\tau}^{(\mathfrak{C})}$.

\begin{theorem}\label{thr-property-evolution-1}
If $\tau\in \Gamma_n$  and $u\in \mathcal{K}$,
then\ $\mathsf{W}\!_{\mathfrak{C}}(\widehat{Z}_{\tau}\otimes u) = \widehat{Z}_{\tau}\otimes \big(U_{\tau}^{(\mathfrak{C})}u\big)$.
\end{theorem}

\begin{proof}
Let $\tau\in \Gamma_n$  and $u\in \mathcal{K}$ be given. Then
$(\partial_k^* +\partial_k)\widehat{Z}_{\tau} = \varepsilon_{\tau}(k)\widehat{Z}_{\tau}$ for all $k\in \mathbb{N}_n$.
Thus, by the definition of the evolution operator $ \mathsf{W}\!_{\mathfrak{C}}$, we have
\begin{equation*}
  \mathsf{W}\!_{\mathfrak{C}}(\widehat{Z}_{\tau}\otimes u)
   = \sum_{k=0}^n [(\partial_k^* +\partial_k)\widehat{Z}_{\tau}]\otimes (C_ku)
   = \sum_{k=0}^n (\varepsilon_{\tau}(k)\widehat{Z}_{\tau})\otimes (C_ku)
   = \widehat{Z}_{\sigma}\otimes\sum_{k=0}^n \varepsilon_{\tau}(k)C_ku,
\end{equation*}
which, together with Definition~\ref{def-weighted-sum-of-coin}, gives
$\mathsf{W}\!_{\mathfrak{C}}(\widehat{Z}_{\tau}\otimes u)= \widehat{Z}_{\tau}\otimes \big(U_{\tau}^{(\mathfrak{C})}u\big)$.
\end{proof}

We are now ready to establish an explicit formula for calculating the probability distribution of the walk at any time $t\geq 0$.

\begin{theorem}\label{prob-distribution-formula}
 For any $t\geq 0$, the probability distribution of the walk $\mathsf{W}\!_{\mathfrak{C}}$ at time $t$ has a representation of the following form
\begin{equation}\label{eq}
  P_t(\sigma\,|\,\Phi_0) = \frac{1}{2^{n+1}}\Big\|\sum_{\tau\in \Gamma_n}(-1)^{\#(\sigma\setminus \tau)} \big(U_{\tau}^{(\mathfrak{C})}\big)^tu_{\tau}\Big\|_{\mathcal{K}}^2,\quad
 \sigma\in \Gamma_n,
\end{equation}
where $u_{\tau}=\sum_{j=0}^{d_{\mathcal{K}}-1} \langle \widehat{Z}_{\tau}\otimes e_j, \Phi_0\rangle e_j$ for $\tau\in \Gamma_n$ and $\Phi_0$ is the initial state of the walk.
\end{theorem}

\begin{proof}
According to Remark~\ref{rem-3-3}, the initial state $\Phi_0$ has an expansion of the following form
\begin{equation}
  \Phi_0 = \sum_{\tau\in \Gamma_n}\widehat{Z}_{\tau}\otimes u_{\tau},
\end{equation}
which together with Theorem~\ref{thr-property-evolution-1} implies that, at time $t\geq 0$, the walk's state $\Phi_t$ can be expressed as
\begin{equation*}
  \Phi_t= {\mathsf{W}\!_{\mathfrak{C}}}^t\Phi_0
   =  \sum_{\tau\in \Gamma_n}{\mathsf{W}\!_{\mathfrak{C}}}^t(\widehat{Z}_{\tau}\otimes u_{\tau})
   = \sum_{\tau\in \Gamma_n}\widehat{Z}_{\tau}\otimes \big(U_{\tau}^{(\mathfrak{C})}\big)^tu_{\tau}.
\end{equation*}
Let $\sigma\in \Gamma_n$ be given. Then, using the above expression, we have
\begin{equation*}
\begin{split}
  P_t(\sigma\,|\,\Phi_0)
  & = \sum_{j=0}^{d_{\mathcal{K}}-1}|\langle Z_{\sigma}\otimes e_j, \Phi_t\rangle|^2\\
  & = \sum_{j=0}^{d_{\mathcal{K}}-1}\Big|\sum_{\tau\in \Gamma_n}
      \big\langle Z_{\sigma}\otimes e_j, \widehat{Z}_{\tau}\otimes \big(U_{\tau}^{(\mathfrak{C})}\big)^tu_{\tau}\big\rangle\Big|^2\\
  & = \sum_{j=0}^{d_{\mathcal{K}}-1}\Big|\sum_{\tau\in \Gamma_n}\langle Z_{\sigma},  \widehat{Z}_{\tau}\rangle_{\mathfrak{h}}
             \big\langle e_j,\big(U_{\tau}^{(\mathfrak{C})}\big)^tu_{\tau}\big\rangle_{\mathcal{K}}\Big|^2\\
  & = \sum_{j=0}^{d_{\mathcal{K}}-1}\Big| \Big\langle e_j,\sum_{\tau\in \Gamma_n}\langle Z_{\sigma},
       \widehat{Z}_{\tau}\rangle_{\mathfrak{h}}\big(U_{\tau}^{(\mathfrak{C})}\big)^tu_{\tau}\Big\rangle_{\mathcal{K}}\Big|^2\\
  & = \Big\|\sum_{\tau\in \Gamma_n}\langle Z_{\sigma},
       \widehat{Z}_{\tau}\rangle_{\mathfrak{h}}\big(U_{\tau}^{(\mathfrak{C})}\big)^tu_{\tau}\Big\|_{\mathcal{K}}^2,
\end{split}
\end{equation*}
which together with
$\langle Z_{\sigma}, \widehat{Z}_{\tau}\rangle_{\mathfrak{h}} = \frac{1}{\sqrt{2^{n+1}}}(-1)^{\#(\sigma\setminus \tau)}$ yields
\begin{equation*}
  P_t(\sigma\,|\,\Phi_0)
    = \Big\|\sum_{\tau\in \Gamma_n}\langle Z_{\sigma}, \widehat{Z}_{\tau}\rangle_{\mathfrak{h}}\big(U_{\tau}^{(\mathfrak{C})}\big)^tu_{\tau}\Big\|_{\mathcal{K}}^2
  = \frac{1}{2^{n+1}}\Big\|\sum_{\tau\in \Gamma_n}(-1)^{\#(\sigma\setminus \tau)}\big(U_{\tau}^{(\mathfrak{C})}\big)^tu_{\tau}\Big\|_{\mathcal{K}}^2.
\end{equation*}
This completes the proof.
\end{proof}

\begin{definition}
For time $T\geq 1$, the $T$-averaged probability distribution of the walk $\mathsf{W}\!_{\mathfrak{C}}$
is defined as
\begin{equation}\label{eq}
  \overline{P}_T(\sigma\,|\,\Phi_0) = \frac{1}{T}\sum_{t=0}^{T-1}P_t(\sigma\,|\,\Phi_0),\quad \sigma \in \Gamma_n,
\end{equation}
where $P_t(\sigma\,|\,\Phi_0)$ is the probability of finding the walker on vertex $\sigma$ at time $t$ and $\Phi_0$ is the initial state.
\end{definition}

The next result establishes a limit theorem for the walk $\mathsf{W}\!_{\mathfrak{C}}$,
which shows that under some mild conditions the walk has a limit averaged probability distribution as time $T$ goes to $\infty$.

\begin{theorem}\label{limit-thr-1}
Suppose that, for each $\tau\in \Gamma_n$, $u_{\tau}=\sum_{j=0}^{d_{\mathcal{K}}-1} \langle \widehat{Z}_{\tau}\otimes e_j, \Phi_0\rangle e_j$
is an eigenvector of $U_{\tau}^{(\mathfrak{C})}$ with $b_{\tau}$ being the corresponding eigenvalue.
Then, on all vertex $\sigma \in \Gamma_n$, one has
\begin{equation}\label{limit-formula-1}
\lim_{T\to \infty}\overline{P}_T(\sigma\,|\,\Phi_0)
= \frac{1}{2^{n+1}}\Big[1+\sum_{(\tau_1,\tau_2)}(-1)^{\#(\sigma\setminus \tau_1) + \#(\sigma\setminus \tau_2)}
  \langle u_{\tau_1}, u_{\tau_2}\rangle_{\mathcal{K}}\Big],
\end{equation}
where $\sum_{(\tau_1,\tau_2)}$ means to sum over the set
$\big\{(\tau_1, \tau_2) \in \Gamma_n\times \Gamma_n \mid \tau_1\neq \tau_2,\, b_{\tau_1}=b_{\tau_2}\big\}$.
\end{theorem}

\begin{proof}
Divide $\Gamma_n\times \Gamma_n$ into three parts $\Gamma_n\times \Gamma_n = \triangle_1\cup \triangle_2 \cup \triangle_3$, where
$\triangle_1=\{(\tau, \tau) \mid \tau\in  \Gamma_n\}$,
\begin{equation*}
\triangle_2=\{(\tau_1, \tau_2) \in \Gamma_n\times \Gamma_n \mid \tau_1\neq \tau_2,\, b_{\tau_1} = b_{\tau_2}\},\quad
\triangle_3=\{(\tau_1, \tau_2) \in \Gamma_n\times \Gamma_n \mid \tau_1\neq \tau_2,\, b_{\tau_1} \neq b_{\tau_2}\}.
\end{equation*}
Let $\sigma \in \Gamma_n$ be given. Then, by Theorem~\ref{prob-distribution-formula}
as well as the equality $\sum_{\tau\in \Gamma_n}\|u_{\tau}\|_{\mathcal{K}}^2=1$, we have
\begin{equation*}
\begin{split}
P_t(\sigma\,|\,\Phi_0)
 &= \frac{1}{2^{n+1}}\Big\|\sum_{\tau\in \Gamma_n}(-1)^{\#(\sigma\setminus \tau)} b_{\tau}^tu_{\tau}\Big\|_{\mathcal{K}}^2\\
 &= \frac{1}{2^{n+1}}\Big[\sum_{(\tau_1,\tau_2)\in \triangle_1} + \sum_{(\tau_1,\tau_2)\in \triangle_2} + \sum_{(\tau_1,\tau_2)\in \triangle_3} \Big]
    (-1)^{\#(\sigma\setminus \tau_1) + \#(\sigma\setminus \tau_2)}\overline{b_{\tau_1}}^tb_{\tau_2}^t\langle u_{\tau_1}, u_{\tau_2}\rangle_{\mathcal{K}}\\
&= \frac{1}{2^{n+1}}\Big[\sum_{\tau\in \Gamma_n}\|u_{\tau}\|_{\mathcal{K}}^2
    +\sum_{(\tau_1,\tau_2)\in \triangle_2}(-1)^{\#(\sigma\setminus \tau_1) + \#(\sigma\setminus \tau_2)}\langle u_{\tau_1}, u_{\tau_2}\rangle_{\mathcal{K}}\\
&\qquad\qquad\ \  +\sum_{(\tau_1,\tau_2)\in \triangle_3}(-1)^{\#(\sigma\setminus \tau_1) + \#(\sigma\setminus \tau_2)}\overline{b_{\tau_1}}^tb_{\tau_2}^t
       \langle u_{\tau_1}, u_{\tau_2}\rangle_{\mathcal{K}}\Big]\\
&= \frac{1}{2^{n+1}}\Big[1
    +\sum_{(\tau_1,\tau_2)\in \triangle_2}(-1)^{\#(\sigma\setminus \tau_1) + \#(\sigma\setminus \tau_2)}\langle u_{\tau_1}, u_{\tau_2}\rangle_{\mathcal{K}}\Big]\\
&\qquad\qquad\ \  +\frac{1}{2^{n+1}}\sum_{(\tau_1,\tau_2)\in \triangle_3}(-1)^{\#(\sigma\setminus \tau_1) + \#(\sigma\setminus \tau_2)}\overline{b_{\tau_1}}^tb_{\tau_2}^t
     \langle u_{\tau_1}, u_{\tau_2}\rangle_{\mathcal{K}},
\end{split}
\end{equation*}
which implies that
\begin{equation*}
  \begin{split}
\overline{P}_T(\sigma\,|\,\Phi_0)
 &= \frac{1}{2^{n+1}}\Big[1+\sum_{(\tau_1,\tau_2)\in \triangle_2}(-1)^{\#(\sigma\setminus \tau_1) + \#(\sigma\setminus \tau_2)}
       \langle u_{\tau_1}, u_{\tau_2}\rangle_{\mathcal{K}}\Big]\\
 &\qquad\qquad\ \  +\frac{1}{2^{n+1}}\sum_{(\tau_1,\tau_2)\in \triangle_3}(-1)^{\#(\sigma\setminus \tau_1) + \#(\sigma\setminus \tau_2)}
       \langle u_{\tau_1}, u_{\tau_2}\rangle_{\mathcal{K}}\frac{1}{T}\sum_{t=0}^{T-1}\overline{b_{\tau_1}}^tb_{\tau_2}^t,
\end{split}
\end{equation*}
where $T\geq 1$. Note that
\begin{equation*}
 \lim_{T\to \infty} \frac{1}{T}\sum_{t=0}^{T-1}\overline{b_{\tau_1}}^tb_{\tau_2}^t
  =  \lim_{T\to \infty}\frac{1}{T}\frac{1-\big(\overline{b_{\tau_1}}b_{\tau_2}\big)^T}{1-\overline{b_{\tau_1}}b_{\tau_2}}
  =0.
\end{equation*}
Thus
\begin{equation*}
  \lim_{T\to \infty} \overline{P}_T(\sigma\,|\,\Phi_0)
   = \frac{1}{2^{n+1}}\Big[1+\sum_{(\tau_1,\tau_2)\in \triangle_2}(-1)^{\#(\sigma\setminus \tau_1) + \#(\sigma\setminus \tau_2)}
     \langle u_{\tau_1}, u_{\tau_2}\rangle_{\mathcal{K}}\Big],
\end{equation*}
which is the same as (\ref{limit-formula-1}).
\end{proof}

\begin{remark}\label{rem-3-4}
For $\gamma\in \Gamma_n$, let $v_{\gamma}$ be an eigenvector of $U_{\gamma}^{(\mathfrak{C})}$ with $b_{\gamma}$ being the corresponding eigenvalue
(such an eigenvector does exist because the coin space $\mathcal{K}$ is finite-dimensional
and $U_{\gamma}^{(\mathfrak{C})}$ is a unitary operator on $\mathcal{K}$).
Put
\begin{equation*}
  \Phi_0= M_0^{-\frac{1}{2}}\sum_{\gamma\in \Gamma_n} \widehat{Z}_{\gamma}\otimes v_{\gamma},
\end{equation*}
where $M_0= \sum_{\gamma\in \Gamma_n}\|v_{\gamma}\|_{\mathcal{K}}^2$. Then, $\Phi_0$ is a unit vector in $\mathfrak{h}_n\otimes \mathcal{K}$,
hence can serve as an initial state of the walk $\mathsf{W}\!_{\mathfrak{C}}$. Moreover, for each $\tau\in \Gamma_n$, we have
\begin{equation*}
  u_{\tau}=\sum_{j=0}^{d_{\mathcal{K}}-1} \langle \widehat{Z}_{\tau}\otimes e_j, \Phi_0\rangle e_j = M_0^{-\frac{1}{2}} v_{\tau},
\end{equation*}
which implies that $u_{\tau}$ is an eigenvector of $U_{\tau}^{(\mathfrak{C})}$ with $b_{\tau}$ being the corresponding eigenvalue.
This shows that the assumptions made in Theorem~\ref{limit-thr-1} can be satisfied.
\end{remark}

As an immediate consequence of Theorem~\ref{limit-thr-1}, the next theorem shows that the limit averaged probability distribution
of the walk $\mathsf{W}\!_{\mathfrak{C}}$ even coincides with the uniform probability distribution on $\Gamma_n$.

\begin{theorem}\label{limit-thr-2}
Suppose that, for each $\tau\in \Gamma_n$, $u_{\tau}=\sum_{j=0}^{d_{\mathcal{K}}-1} \langle \widehat{Z}_{\tau}\otimes e_j, \Phi_0\rangle e_j$
is an eigenvector of $U_{\tau}^{(\mathfrak{C})}$ with $b_{\tau}$ being the corresponding eigenvalue,
and moreover $b_{\tau_1}\neq b_{\tau_2}$ when $\tau_1$ $\tau_2\in \Gamma_n$ with $\tau_1\neq \tau_2$.
Then, it holds true that
\begin{equation}\label{limit-formula-2}
\lim_{T\to \infty}\overline{P}_T(\sigma\,|\,\Phi_0)
= \frac{1}{2^{n+1}},\quad  \forall\, \sigma \in \Gamma_n.
\end{equation}
\end{theorem}

\begin{proof}
It follows from the conditions that $\big\{(\tau_1, \tau_2) \in \Gamma_n\times \Gamma_n \mid \tau_1\neq \tau_2,\, b_{\tau_1}=b_{\tau_2} \big\}=\emptyset$,
which together with Theorem~\ref{limit-thr-1} implies (\ref{limit-formula-2}).
\end{proof}

\begin{remark}\label{rem-3-5}
The unitary operators $\big\{U_{\tau}^{(\mathfrak{C})} \mid \tau \in \Gamma_n\big\}$ can be viewed as the evolution operator's ``components in the coin space''.
Similarly, the vectors $\{u_{\tau} \mid \tau \in \Gamma_n\big\}$ can be viewed as the initial state's ``components in the coin space''.
Theorem~\ref{prob-distribution-formula}, Theorem~\ref{limit-thr-1} and Theorem~\ref{limit-thr-2}
suggest that the walk's ``components in the coin space'' are actually the determining factors of
its probability distributions.
\end{remark}

\subsection{The uniform measure as a stationary measure}\label{subsec-3-5}

The present subsection shows that the walk $\mathsf{W}\!_{\mathfrak{C}}$ produces the uniform measure as its stationary measure on $\Gamma_n$
provided its initial state satisfies some mild conditions.

In this subsection, we assume that $n\geq 0$ is a fixed nonnegative integer and the coin space $\mathcal{K}$ is a finite-dimensional Hilbert space with $d_{\mathcal{K}}\equiv\dim \mathcal{K}\geq 2^{n+1}$.
Additionally, we assume that that $\mathfrak{C}=\{C_k \mid 0\leq k \leq n\}$ is a given coin operator system on $\mathcal{K}$.
As above, $\Phi_0$ always denotes the initial state of the walk $\mathsf{W}\!_{\mathfrak{C}}$ below.

A probability measure (measure for short) $\nu$ on $\Gamma_n$ is a nonnegative function $\nu\colon \Gamma_n \rightarrow \mathbb{R}_+$ satisfying
that $\sum_{\sigma\in \Gamma_n}\nu(\sigma)=1$. The uniform measure $\kappa_n$ on $\Gamma_n$ is the measure given by
\begin{equation}\label{eq}
  \kappa_n(\sigma) = \frac{1}{2^{n+1}},\quad \sigma\in \Gamma_n.
\end{equation}

\begin{definition}
A measure $\nu$ on $\Gamma_n$ is called a stationary measure of the walk $\mathsf{W}\!_{\mathfrak{C}}$ if there exists
some unit vector $\Psi\in \mathfrak{h}_n\otimes \mathcal{K}$ such that
\begin{equation}\label{eq}
  \nu(\sigma)= \sum_{j=0}^{d_{\mathcal{K}}-1} |\langle Z_{\sigma}\otimes e_j, {\mathsf{W}\!_{\mathfrak{C}}}^t\Psi\rangle|^2,\quad
  \forall\, \sigma\in \Gamma_n,\, \forall\, t\geq 0.
\end{equation}
Furthermore, if the stationary measure $\nu$ coincides with the uniform measure $\kappa_n$,
then we say that the walk $\mathsf{W}\!_{\mathfrak{C}}$ with $\Phi_0=\Psi$ produces the uniform measure as its stationary measure on $\Gamma_n$.
\end{definition}

\begin{theorem}\label{thr-stationary-distribution-1}
For all $\gamma\in \Gamma_n$ and all $u\in \mathcal{K}$ with $\|u\|_{\mathcal{K}}^2=1$,
the walk $\mathsf{W}\!_{\mathfrak{C}}$ with $\Phi_0=\widehat{Z}_{\gamma}\otimes u$ produces the uniform measure as its stationary measure on $\Gamma_n$.
In particular, the uniform measure $\kappa_n$ on $\Gamma_n$ is a stationary measure of the walk $\mathsf{W}\!_{\mathfrak{C}}$.
\end{theorem}

\begin{proof}
Let $\gamma\in \Gamma_n$ and $u\in \mathcal{K}$ with $\|u\|_{\mathcal{K}}^2=1$ be given. Then, for all $\sigma\in \Gamma_n$ and $t\geq 0$,
by Theorem~\ref{thr-property-evolution-1} we have
\begin{equation*}
\begin{split}
 \sum_{j=0}^{d_{\mathcal{K}}-1} \big|\big\langle Z_{\sigma}\otimes e_j, {\mathsf{W}\!_{\mathfrak{C}}}^t(\widehat{Z}_{\gamma}\otimes u)\big\rangle\big|^2
  & =\sum_{j=0}^{d_{\mathcal{K}}-1} |\langle Z_{\sigma}\otimes e_j, \widehat{Z}_{\gamma}\otimes \big(U_{\gamma}^{(\mathfrak{C})}\big)^tu\rangle|^2\\
  & =\sum_{j=0}^{d_{\mathcal{K}}-1} \big|\langle Z_{\sigma}, \widehat{Z}_{\gamma}\rangle_{\mathfrak{h}}
      \big\langle e_j, \big(U_{\gamma}^{(\mathfrak{C})}\big)^tu\big\rangle_{\mathcal{K}}\big|^2\\
  & =|\langle Z_{\sigma}, \widehat{Z}_{\gamma}\rangle_{\mathfrak{h}}|^2\big\|(U_{\gamma}^{(\mathfrak{C})}\big)^tu\big\|_{\mathcal{K}}^2\\
  & = |\langle Z_{\sigma}, \widehat{Z}_{\gamma}\rangle_{\mathfrak{h}}|^2\|u\|_{\mathcal{K}}^2\\
  & = |\langle Z_{\sigma}, \widehat{Z}_{\gamma}\rangle_{\mathfrak{h}}|^2,
\end{split}
\end{equation*}
which together with $\langle Z_{\sigma}, \widehat{Z}_{\gamma}\rangle_{\mathfrak{h}} =\frac{(-1)^{\#(\sigma\setminus \gamma)}}{\sqrt{2^{n+1}}}$ gives
\begin{equation*}
   \sum_{j=0}^{d_{\mathcal{K}}-1} \big|\big\langle Z_{\sigma}\otimes e_j, {\mathsf{W}\!_{\mathfrak{C}}}^t(\widehat{Z}_{\gamma}\otimes u)\big\rangle\big|^2
   =|\langle Z_{\sigma}, \widehat{Z}_{\gamma}\rangle_{\mathfrak{h}}|^2
   = \frac{1}{2^{n+1}}.
\end{equation*}
This exactly means that walk $\mathsf{W}\!_{\mathfrak{C}}$ with $\Phi_0=\widehat{Z}_{\gamma}\otimes u$ produces the uniform measure
as its stationary measure on $\Gamma_n$.
\end{proof}


Recall that, for $\tau\in \Gamma_n$, $U_{\tau}^{(\mathfrak{C})}$ is the $\varepsilon_{\tau}$-weighted sum of the coin operator system $\mathfrak{C}$
(see Definition~\ref{def-weighted-sum-of-coin}).
The next theorem shows that,
even for some ``complicated'' initial states, the walk $\mathsf{W}\!_{\mathfrak{C}}$ still produces the uniform measure.

\begin{theorem}\label{thr-unifotm-distribution}
Let $\Psi\in \mathfrak{h}_n\otimes \mathcal{K}$ be a unit vector
and $u_{\tau}=\sum_{j=0}^{d_{\mathcal{K}}-1} \langle \widehat{Z}_{\tau}\otimes e_j, \Psi\rangle e_j$ for $\tau\in \Gamma_n$.
Suppose further that
\begin{enumerate}
  \item[(1)] for each $\tau\in \Gamma_n$, $u_{\tau}$ is an eigenvector of the unitary operator $U_{\tau}^{(\mathfrak{C})}$;
  \item[(2)] $\langle u_{\tau_1},u_{\tau_2} \rangle_{\mathcal{K}}=0$ for $\tau_1$, $\tau_2\in \Gamma_n$ with $\tau_1\neq \tau_2$.
\end{enumerate}
Then the walk $\mathsf{W}\!_{\mathfrak{C}}$ with $\Phi_0=\Psi$ produces the uniform measure as its stationary measure on $\Gamma_n$.
\end{theorem}

\begin{proof}
Let $\Phi_0=\Psi$ and denote by $b_{\tau}$ the eigenvalue corresponding to the eigenvector $u_{\tau}$.
Then, for all $\sigma\in \Gamma_n$ and $t\geq 0$, by Theorem~\ref{prob-distribution-formula} we immediately have
\begin{equation*}
 \sum_{j=0}^{d_{\mathcal{K}}-1} |\langle Z_{\sigma}\otimes e_j, {\mathsf{W}\!_{\mathfrak{C}}}^t\Psi\rangle|^2
  = P_t(\sigma\,|\,\Phi_0)
  = \frac{1}{2^{n+1}}\Big\|\sum_{\tau\in \Gamma_n}(-1)^{\#(\sigma\setminus \tau)} b_{\tau}^tu_{\tau}\Big\|_{\mathcal{K}}^2
  = \frac{1}{2^{n+1}}\sum_{\tau\in \Gamma_n}\|u_{\tau}\|_{\mathcal{K}}^2,
\end{equation*}
which together with
\begin{equation*}
  \sum_{\tau\in \Gamma_n}\|u_{\tau}\|_{\mathcal{K}}^2
  = \sum_{\tau\in \Gamma_n}\sum_{j=0}^{d_{\mathcal{K}}-1} |\langle \widehat{Z}_{\tau}\otimes e_j, \Psi\rangle|^2
  = \|\Psi\|^2
  = 1
\end{equation*}
yields
\begin{equation*}
  \sum_{j=0}^{d_{\mathcal{K}}-1} |\langle Z_{\sigma}\otimes e_j, {\mathsf{W}\!_{\mathfrak{C}}}^t\Psi\rangle|^2
 = \frac{1}{2^{n+1}}\sum_{\tau\in \Gamma_n}\|u_{\tau}\|_{\mathcal{K}}^2
 = \frac{1}{2^{n+1}}.
\end{equation*}
This exactly means that the walk $\mathsf{W}\!_{\mathfrak{C}}$ with $\Phi_0=\Psi$ produces the uniform measure as its stationary measure on $\Gamma_n$.
\end{proof}

\subsection{Examples}\label{subsec-example}

In the final subsection, we offer some examples to show that the assumptions made in Theorem~\ref{limit-thr-2} and Theorem~\ref{thr-unifotm-distribution} can be satisfied.

Consider the graph $(\Gamma_1, \mathfrak{E}_1)$, which is actually isomorphic to the $2$-dimensional hypercube.
In this case, the position space of the walk is just
$\mathfrak{h}_1 = \mathrm{span}\big\{Z_{\sigma}\mid \sigma \in \Gamma_1\big\}$.
Note that
\begin{equation}\label{eq-3-29}
\Gamma_1 =\big\{\emptyset, \{0\}, \{1\}, \{0,1\}\big\}.
\end{equation}
Additionally, the graph $(\Gamma_1, \mathfrak{E}_1)$ is regular and its degree is exactly $2$.

\begin{example}\label{exam-3-1}
Take $\mathbb{C}^2$ as the coin space, namely $\mathcal{K}=\mathbb{C}^2$. Then, one has a coin operator system
$\mathfrak{C}=\{C_0, C_1\}$ on $\mathcal{K}=\mathbb{C}^2$,
where
\begin{equation}\label{eq-3-30}
  C_0 =
    \left(
    \begin{matrix}
    0 & 1\\
    0 & 0
   \end{matrix}
   \right),\quad
  C_1=
    \left(
    \begin{matrix}
    0 & 0\\
    1 & 0
   \end{matrix}
   \right).
\end{equation}
In this case, the state space of the walk is
$\mathfrak{h}_1\otimes \mathcal{K} = \mathfrak{h}_1\otimes \mathbb{C}^2$, while the evolution operator $\mathsf{W}\!_{\mathfrak{C}}$
takes the form
\begin{equation}\label{eq-3-31}
  \mathsf{W}\!_{\mathfrak{C}}
    = (\partial_0^*+\partial_0)\otimes C_0 + (\partial_1^*+\partial_1)\otimes C_1.
\end{equation}
According to Definition~\ref{def-weighted-sum-of-coin}, the weighted sums of the coin operator system $\mathfrak{C}$ are
\begin{equation*}
U_{\emptyset}^{(\mathfrak{C})}=
  \left(
    \begin{matrix}
    0 & -1\\
    -1 & 0
   \end{matrix}
   \right),\quad
U_{\{0\}}^{(\mathfrak{C})}=
  \left(
    \begin{matrix}
    0 & 1\\
    -1 & 0
   \end{matrix}
   \right),\quad
U_{\{1\}}^{(\mathfrak{C})}=
  \left(
    \begin{matrix}
    0 & -1\\
    1 & 0
   \end{matrix}
   \right),\quad
U_{\{0,1\}}^{(\mathfrak{C})}=
  \left(
    \begin{matrix}
    0 & 1\\
    1 & 0
   \end{matrix}
   \right).
\end{equation*}
By careful calculations, one can immediately get spectrums of these unitary operators (matrices):
\begin{equation*}
  \mathrm{Spec}\big(U_{\emptyset}^{(\mathfrak{C})}\big) = \mathrm{Spec}\big(U_{\{0,1\}}^{(\mathfrak{C})}\big) =\{-1, 1\},\quad
  \mathrm{Spec}\big(U_{\{0\}}^{(\mathfrak{C})}\big) = \mathrm{Spec}\big(U_{\{1\}}^{(\mathfrak{C})}\big) =\{-\mathrm{i}, \mathrm{i}\},
\end{equation*}
where $\mathrm{Spec}(A)$ means the spectrum of an operator $A$ acting on $\mathcal{K}=\mathbb{C}^2$.

Let $v_{\emptyset}\in \mathbb{C}^2$ be an eigenvector of $U_{\emptyset}^{(\mathfrak{C})}$ with $b_{\emptyset} = -1$ being the corresponding eigenvalue,
$v_{\{0\}}\in \mathbb{C}^2$ be an eigenvector of $U_{\{0\}}^{(\mathfrak{C})}$ with $b_{\{0\}} = -\mathrm{i}$ being the corresponding eigenvalue,
$v_{\{1\}}\in \mathbb{C}^2$ be an eigenvector of $U_{\{1\}}^{(\mathfrak{C})}$ with $b_{\{1\}} = \mathrm{i}$ being the corresponding eigenvalue,
and $v_{\{0,1\}}\in \mathbb{C}^2$ be an eigenvector of $U_{\{0,1\}}^{(\mathfrak{C})}$ with $b_{\{0,1\}} = 1$ being the corresponding eigenvalue.
Put
\begin{equation*}
  \Phi_0= M_0^{-\frac{1}{2}}\sum_{\gamma\in \Gamma_1} \widehat{Z}_{\gamma}\otimes v_{\gamma},
\end{equation*}
where $M_0= \sum_{\gamma\in \Gamma_1}\|v_{\gamma}\|_{\mathbb{C}^2}^2$. Then, $\Phi_0$ is a unit vector in $\mathfrak{h}_1\otimes \mathcal{K}= \mathfrak{h}_1\otimes \mathbb{C}^2$, hence can serve as an initial state. Furthermore, for each $\tau \in \Gamma_1=\big\{\emptyset, \{0\}, \{1\}, \{0,1\}\big\}$, we have
\begin{equation*}
  u_{\tau}=\sum_{j=0}^{d_{\mathcal{K}}-1} \langle \widehat{Z}_{\tau}\otimes e_j, \Phi_0\rangle e_j = M_0^{-\frac{1}{2}}v_{\tau},
\end{equation*}
which implies that $U_{\tau}^{(\mathfrak{C})}u_{\tau}=b_{\tau} u_{\tau}$, namely $u_{\tau}$ is an eigenvector of $U_{\tau}^{(\mathfrak{C})}$
with $b_{\tau}$ being the corresponding eigenvalue. Clearly, $b_{\tau_1}\neq b_{\tau_2}$ for $\tau_1$, $\tau_1\in \Gamma_1$
with $\tau_1\neq \tau_2$.
This shows that the assumptions made in Theorem~\ref{limit-thr-2} can be satisfied.
\end{example}

\begin{example}\label{exam-3-2}
Take $\mathbb{C}^4$ as the coin space, namely $\mathcal{K}=\mathbb{C}^4$. Then, one also has a coin operator system
$\mathfrak{C}=\{C_0, C_1\}$ on $\mathcal{K}=\mathbb{C}^4$,
where
\begin{equation}\label{eq-3-32}
  C_0 =
    \left(
    \begin{matrix}
    1 & 0 & 0 & 0\\
    0 & 1 & 0 & 0\\
    0 & 0 & 0 & 0\\
    0 & 0 & 0 & 0
   \end{matrix}
   \right),\quad
  C_1=
    \left(
    \begin{matrix}
    0 & 0 & 0 & 0\\
    0 & 0 & 0 & 0\\
    0 & 0 & -1 & 0\\
    0 & 0 & 0 & -1
   \end{matrix}
   \right).
\end{equation}
Thus, the state space of the walk is
$\mathfrak{h}_1\otimes \mathcal{K} = \mathfrak{h}_1\otimes \mathbb{C}^4$, and the evolution operator $\mathsf{W}\!_{\mathfrak{C}}$ takes the form
\begin{equation}\label{eq}
  \mathsf{W}\!_{\mathfrak{C}}
    = (\partial_0^*+\partial_0)\otimes C_0 + (\partial_1^*+\partial_1)\otimes C_1.
\end{equation}
By Definition~\ref{def-weighted-sum-of-coin}, one can get the weighted sums of the coin operator system $\mathfrak{C}$ as follows:
\begin{equation*}
U_{\emptyset}^{(\mathfrak{C})}=
  \left(
    \begin{matrix}
    -1 & 0 & 0 & 0\\
    0 & -1 & 0 & 0\\
    0 & 0 & 1 & 0\\
    0 & 0 & 0 & 1
   \end{matrix}
   \right),\quad
U_{\{0\}}^{(\mathfrak{C})}=
  \left(
    \begin{matrix}
    1 & 0 & 0 & 0\\
    0 & 1 & 0 & 0\\
    0 & 0 & 1 & 0\\
    0 & 0 & 0 & 1
   \end{matrix}
   \right),
\end{equation*}
\begin{equation*}
U_{\{1\}}^{(\mathfrak{C})}=
  \left(
    \begin{matrix}
    -1 & 0 & 0 & 0\\
    0 & -1 & 0 & 0\\
    0 & 0 & -1 & 0\\
    0 & 0 & 0 & -1
   \end{matrix}
   \right),\quad
U_{\{0,1\}}^{(\mathfrak{C})}=
  \left(
    \begin{matrix}
    1 & 0 & 0 & 0\\
    0 & 1 & 0 & 0\\
    0 & 0 & -1 & 0\\
    0 & 0 & 0 & -1
   \end{matrix}
   \right),
\end{equation*}
which are unitary operators (matrices) on $\mathcal{K}=\mathbb{C}^4$.
These operators respectively have eigenvectors $v_{\emptyset}$, $v_{\{0\}}$, $v_{\{1\}}$, $v_{\{0,1\}}$ given by
\begin{equation*}
  v_{\emptyset} =\big(0, 0,\frac{1}{\sqrt{2}}, \frac{1}{\sqrt{2}} \big)^T,\quad v_{\{0\}}= \big(0,0,\frac{1}{\sqrt{2}},\frac{-1}{\sqrt{2}} \big)^T,
\end{equation*}
\begin{equation*}
  v_{\{1\}} = \big(\frac{1}{\sqrt{2}}, \frac{1}{\sqrt{2}}, 0, 0\big)^T,\quad
  v_{\{0,1\}}=\big(\frac{1}{\sqrt{2}},\frac{-1}{\sqrt{2}},0,0\big)^T.
\end{equation*}
Clearly, $\{v_{\emptyset}, v_{\{0\}}, v_{\{1\}}, v_{\{0,1\}}\}$ forms an orthonormal system in $\mathcal{K}=\mathbb{C}^4$. Now put
\begin{equation}\label{eq}
  \Psi = \frac{1}{2}\sum_{\gamma\in \Gamma_1}\widehat{Z}_{\gamma}\otimes v_{\gamma}.
\end{equation}
Then, $\Psi$ is a unit vector in $\mathfrak{h}_1\otimes \mathcal{K}= \mathfrak{h}_1\otimes \mathbb{C}^4$.
Moreover, for each $\tau\in \Gamma_1=\big\{\emptyset, \{0\}, \{1\}, \{0,1\}\big\}$, we have
\begin{equation*}
  u_{\tau} = \sum_{j=0}^{d_{\mathcal{K}}-1} \langle \widehat{Z}_{\tau}\otimes e_j, \Psi\rangle e_j = \frac{1}{2}v_{\tau},
\end{equation*}
which implies that $u_{\tau}$ is an eigenvector of $U_{\tau}^{(\mathfrak{C})}$.
Clearly, for $\tau_1$, $\tau_2\in \Gamma_1$ with $\tau_1\neq \tau_2$, one has
\begin{equation*}
\langle u_{\tau_1},u_{\tau_2} \rangle_{\mathcal{K}}
= \frac{1}{4}\langle v_{\tau_1},v_{\tau_2} \rangle_{\mathcal{K}}=0.
\end{equation*}
This shows that the assumptions made in Theorem~\ref{thr-unifotm-distribution} can also be satisfied.
\end{example}

\section{Conclusion remarks}\label{sec-4}

As is seen, we obtain an alternative description of the $(n+1)$-dimensional hypercube.
And based on the alternative description, we find that the operators $\{\partial_k^* + \partial_k \mid 0\leq k \leq n\}$ behave actually as
the shift operators. This allows us to introduce a quantum walk model on the $(n+1)$-dimensional hypercube with $\{\partial_k^* + \partial_k \mid 0\leq k \leq n\}$ as the shift operators on the position space. We explicitly obtain a formula for calculating the probability distribution of the walk at any time
and establish two limit theorems showing that the averaged probability distribution of the walk even converges to the uniform probability distribution on the
$(n+1)$-dimensional hypercube. Finally, we prove that the walk produces the uniform measure as its stationary measure
on the $(n+1)$-dimensional hypercube provided its initial state satisfies some mild conditions.

From our work, we may come to some observations as follows:
(1) QBN (quantum Bernoulli noises) can provide a useful framework for producing the uniform measure on a general hypercube via a quantum walk.
(2) For a quantum walk on a hypercube, its ``components in the coin space'' can be the determining factors of its probability distributions and evolution behavior.

\section*{Acknowledgements}

The author is extremely grateful to the referees for their valuable comments and suggestions on improvement
of the first version of the present paper.

\section*{Author declarations}

\subsection*{Conflict of interest}

The author has no conflicts to disclose.

\subsection*{Data Availability}

Data sharing is not applicable to this article as no new data were created or analyzed in this study.

\end{document}